\documentclass[sn-mathphys-num]{sn-jnl}

\usepackage{graphicx}%
\usepackage{multirow}%
\usepackage{amsmath,amssymb,amsfonts}%
\usepackage{amsthm}%
\usepackage{mathrsfs}%
\usepackage[title]{appendix}%
\usepackage{xcolor}%
\usepackage{textcomp}%
\usepackage{manyfoot}%
\usepackage{booktabs}%
\usepackage{algorithm}%
\usepackage{algorithmicx}%
\usepackage{algpseudocode}%
\usepackage{listings}%
\usepackage{natbib}

\usepackage{mathtools}

\usepackage{wrapfig}

\newtheorem{theorem}{Theorem}
\newtheorem{proposition}{Proposition}

\newtheorem{definition}{Definition}%

\newtheorem{lemma}{Lemma}
\newtheorem{corollary}[theorem]{Corollary}

\newcommand{\HH}{\mathcal{H}}

\newcommand{\PP}{\mathbb{P}}
\newcommand{\paren}[1]{\left( #1 \right)}

\begin{document}

\title[Article Title]{Hadwiger Models: Low-Temperature Behavior in a Natural Extension of the Ising Model}


\author*[1]{\fnm{Summer} \sur{Eldridge}}\email{seldridge@gradcenter.cuny.edu}

\author[2]{\fnm{Benjamin} \sur{Schweinhart}}\email{bschwei@gmu.edu}

\affil[1]{\orgdiv{Mathematics}, \orgname{Graduate Center, City University of New York}, \orgaddress{\street{365 5th Ave}, \city{New York}, \postcode{10016}, \state{New York}, \country{United States}}}
\affil[2]{\orgdiv{Mathematics}, \orgname{George Mason University}, \orgaddress{\street{4400 University Dr}, \city{Fairfax}, \postcode{22030}, \state{Virginia}, \country{United States}}}


\abstract{In two dimensions, all isometrically invariant Markov random fields on binary assignments are induced by energy functions that can be represented as linear combinations of area, perimeter, and Euler characteristic. This class of model includes the Ising model, both ferro- and antiferromagnetic, with and without a field, as well as the Baxter-Wu model. On the hexagonal lattice, we determine the low-temperature behavior for this class of model, and construct a phase diagram of said behavior. In particular, we identify regions with three geometric phases, regions with a single unique phase, and coexistence curves between them. We also characterize the behavior along two non-Peierls lines, where entropy fails to vanish the as temperature goes to zero.}




\maketitle

\section{Introduction}

Traditionally, the Ising Hamiltonian is defined as a function $H:\{-1,1\}^{V}\to\mathbb{R}$ by $H(\sigma)=\sum_{(x,y)\in E}\sigma(x)\sigma(y)+h\sum_{x\in V}\sigma(x)$, where $V$ and $E$ are the vertex and edge sets of some graph $G$. However, existence of a phase transition for the planar Ising model was first proven by Peierls \cite{P36} using an equivalent formulation where spins are assigned to the faces $F$ of the dual graph $G^*$ of $G,$ and $H$ is defined in terms of the area and perimeter of the union of dual faces assigned spin $1$. Formally, given a spin assignment $\sigma\in \{-1,1\}^{F}$ denote by $\sigma^*$, represented in Fig.~\ref{fig:conversion}, the collection of dual faces assigned spin $1$ and set $H(\sigma^*)=P(\sigma^*)+hA(\sigma^*)$, where $P(\sigma^*)$ and $A(\sigma^*)$ are respectively the area and perimeter of $\sigma.$
This formulation allowed Peierls to bound the perimeter weights with a geometric argument and thus, at $h=0$, guarantee that one orientation dominates at sufficiently low temperatures.

\begin{figure}
    \centering
    \includegraphics[width=0.7\linewidth]{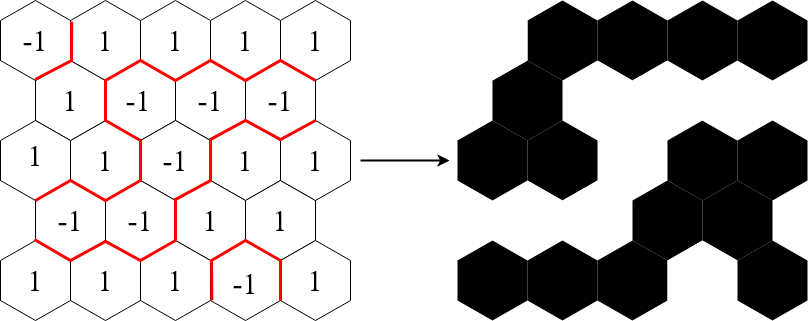}
    
    \caption{Converting binary assignments on a lattice to polyconvex subsets on the faces of the dual lattice. $\sigma^*$ is illustrated by the union of the black hexagons on the right.}
    \label{fig:conversion}
    
\end{figure}

A \emph{valuation} is a function $H$ assigning real numbers to polyconvex subsets of $\mathbb{R}^d$ (where a polyconvex set is one that is a finite union of bounded convex sets) such that $H(\varnothing)=0$ and
$H(U\cup V)=H(U)+H(V)-H(U\cap V)$. Functions satisfying the latter equation have the \emph{Markov property}. If $H$ also satisfies $H(T(U))=H(U)$ for all isometries $T$, then $H$ is called an \emph{invariant valuation}. The Markov property implies that the Gibbs distribution induced by $H$ within a given region depends only on its boundary; if the configuration space assigns positive probability to all possible assignments from the set to the lattice, such a distribution is called a Markov random field. For an infinite planar graph with finite fundamental domain, the dual faces are always compact and polygonal, meaning their union is polyconvex, so any isometrically invariant Markov Hamiltonian on binary assignments is an invariant valuation, up to a constant term for the energy of the empty set. Hadwiger's theorem describes the set of all such functions explicitly:
\begin{theorem}[Hadwiger's Theorem\cite{H56}]
In two dimensions, every invariant valuation must be a linear combination of perimeter, area, and Euler characteristic.
\end{theorem} Recall that in two dimensions Euler characteristic $\chi$ can be expressed as $$\chi=\#\text{vertices}-\#\text{edges}+\#\text{faces}.$$

 All Markov random fields are induced by Markov Hamiltonians \cite{CH71,G73}, so the class of Gibbs distributions induced by invariant valuations coincides with the class of  Markov random fields on binary assignments. \footnote{The relation between Markov Hamiltonians and Markov random fields fails to hold without the positivity condition: for extensions of the statement, see \cite{CM16}.} We summarize the preceding discussion in a corollary.
 
\begin{corollary}
A Markov random field in two dimensions is the Gibbs distribution induced by a Hamiltonian that is a linear combination of perimeter, area, and Euler characteristic. 
\end{corollary} 
Note that each term can be expressed in terms of a product over a set of spins, so this class forms a subfamily within the generalized Ising models. Hadwiger's theorem and the previous discussion extend to $d$-dimensions, where perimeter, area, and Euler characteristic are replaced by a collection of $d+1$ distinguished functions called intrinsic volumes. See~\cite{G07} for a definition. We leave the study of Markov random fields in $d$-dimensions to future work.

This paper is dedicated to the study of the distributions induced by members of this class of Hamiltonian on the hexagonal lattice.  We choose to work on the hexagonal lattice because each term behaves nicely under a spin flip: if a set of faces is replaced with its complement, the perimeter remains unchanged, whereas the area and the Euler characteristic are inverted, up to an additive constant.

We introduce some notation. Define a domain as a subset $\Lambda$ of the faces of a lattice such that $\cup_{f\in \Lambda}f$ and $\mathbb{R}^2\setminus \cup_{f\in \Lambda}f$ are both connected, a configuration $\sigma\in \{-1,1\}^\Lambda$ on a domain as a binary assignment to the faces of $\Lambda$, a boundary configuration $\rho$ as a binary assignment to all the faces in the lattice, an occupied or filled hexagon $f$ as one such that $\sigma(f)=1$ if $f\in \Lambda$ or $\rho(f)=1$ if $f\notin\Lambda$ and, $\partial \Lambda$ (the ``boundary") as the set of faces neighboring $\Lambda$. We identify $\sigma$ and $\rho$ with their associated set of occupied hexagons.

\begin{definition}[Hadwiger Model]\label{def:hadwiger}
    Given a domain $\Lambda$, a boundary configuration $\rho$, and a temperature $T$, define a probability distribution on the set of configurations $\Omega=\{-1,1\}^\Lambda$ with parameters $x,p,a$ by 
    \begin{equation}
    \mu(\sigma)\coloneqq \mu_{x,p,a,T}(\sigma)\propto e^{\mathcal{H}(\sigma)/T},\; \mathcal{H}(\sigma)=x\chi(\sigma)+pP(\sigma)+aA(\sigma)
    \end{equation}
    where $\chi,P,A$ are Euler characteristic, perimeter, and area respectively. 
\end{definition}

While we define the Hadwiger models in terms of the perimeter, area, and Euler functions, the phase diagram is naturally parametrized by vertex energies. As the Hamiltonian is a sum of interactions between neighboring hexagons, we can rewrite it as a sum over the vertices of functions depending only on the states of neighboring hexagons.  Note that every model assigns the empty state energy 0, and that only number and not position of adjacent hexagons affect the energy, because the Hamiltonian is invariant under rotations. Refer to the vertex states with 0,1,2, and 3 incident filled hexagons as $E,C,H,F$ respectively. We denote the energy assigned by the Hamiltonian to vertex state $X$ with $e_X$. Assuming translation and rotation invariance, these energies are uniquely determined by the equation  $\mathcal{H}=\sum_{v\in V} e_{X\paren{v}}.$


Euler characteristic, perimeter, and area are signed sums over the number of adjacent faces, edges, or vertices, so for each element we sum over, we can divide by the number of vertices incident to that element to determine the energy of each vertex in a particular state. For example, the $H$ state vertex has two incident hexagons, three incident edges, and one incident vertex, so it contributes $1=\frac{2}{2}$ to the perimeter, $\frac{1}{3}=\frac{2}{6}$ to the area, and $\frac26-\frac32+1=-\frac16$ to the Euler characteristic. Applying this logic to every state, we determine the vertex energies are
\begin{equation}
\label{eq:vertexenergies}
e_C=\frac{x}{6}+p+\frac{a}{6},\quad e_H=-\frac{x}{6}+p+\frac{a}{3},\quad e_F=\frac{a}{2}\,.
\end{equation}
These assignments, represented in Fig.~\ref{fig:assignments}, span all choices of $e_C,e_H,e_F$, so any choice will be an invariant valuation. Combining this with the previous corollary yields the next result.

\begin{figure}[t]
\begin{minipage}{0.35\linewidth}
    \centering
    \includegraphics[width=\linewidth]{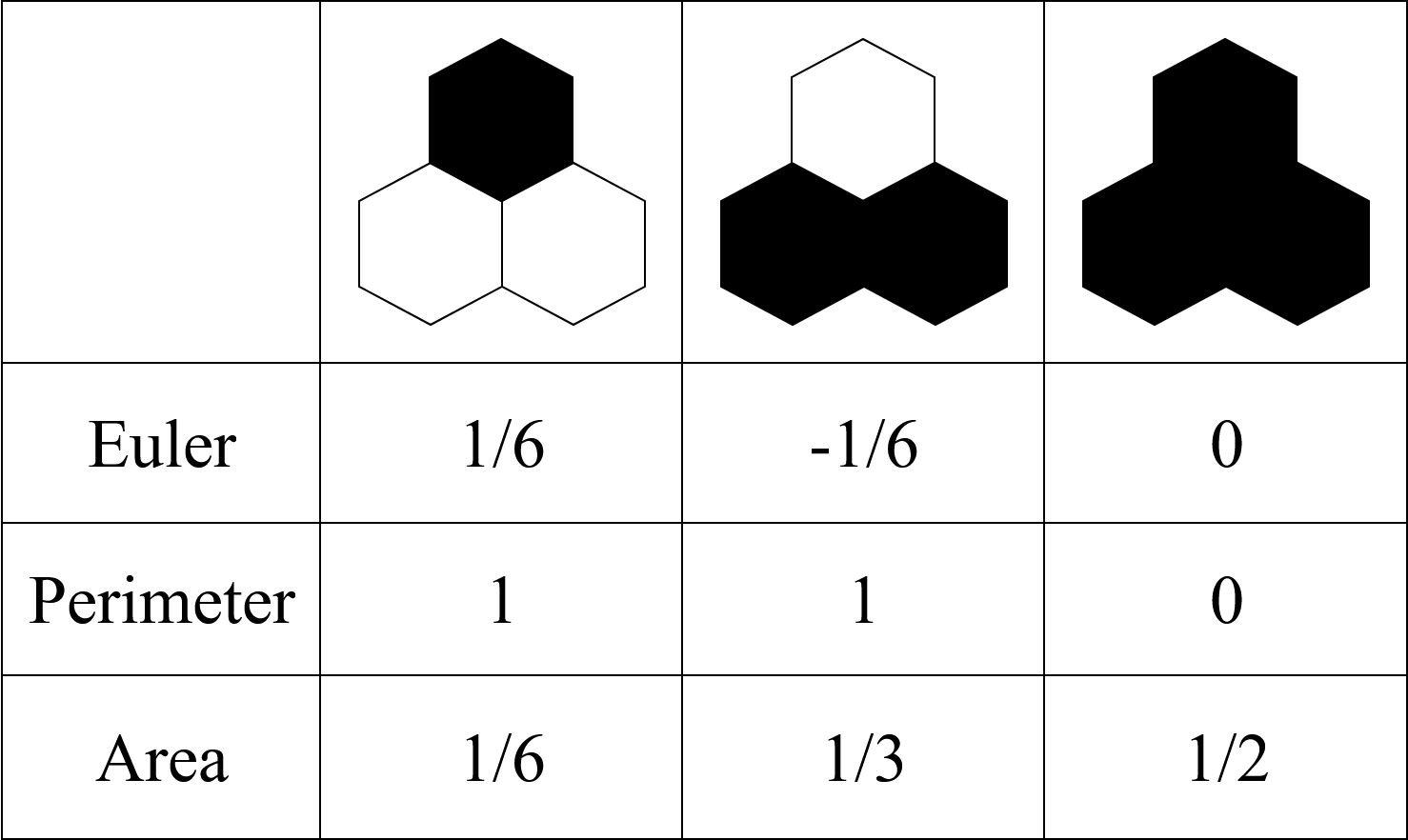}
    
    \caption{Energy assigned by each term to the three non-empty vertex states}
    \label{fig:assignments}
\end{minipage}
    \hspace{\fill}%
\begin{minipage}{0.6\linewidth}
    \centering
    \includegraphics[width=\linewidth]{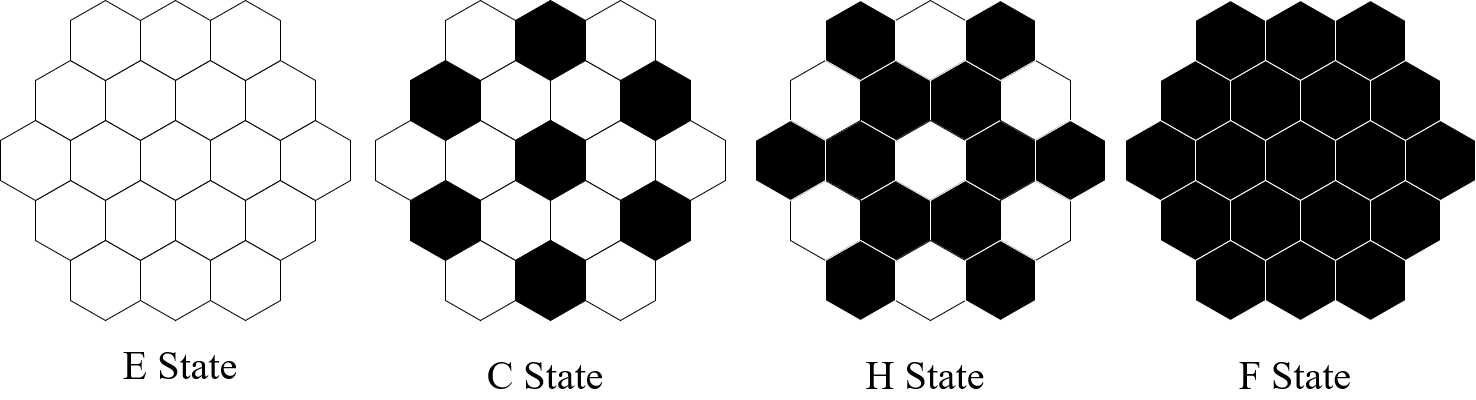}
    \caption{Representative subsets of each of the configurations with exactly 1 vertex state}
    \label{fig:groundstates}
\end{minipage}
\end{figure}

\begin{proposition}
The set of Markov Hamiltonians on the hexagonal lattice is \begin{equation}
\{g(\sigma)=\sum_{v\in V}f(\sigma_v)|f(\sigma_v)=e_C\cdot \delta_{C,\sigma_v}+e_H\cdot \delta_{H,\sigma_v}+e_F\cdot \delta_{F,\sigma_v}\}
\end{equation}
where $\sigma_v$ is the vertex state induced by $\sigma$ on $v$ and $\delta$ is the Kronecker delta function.
\end{proposition}
In particular, we can obtain the Hamiltonian with parameters $x,p,a$ by using the vertex energies given in equation~\ref{eq:vertexenergies}. On the other hand, we can convert back to $x, p, a$ parameters above with the equations
$$x=3e_C-3e_H+e_F,\quad p=\frac{e_C+e_H-e_F}2,\quad a=2e_F.$$

Similar reparametrizations exist for Markov Hamiltonians on other lattices and in higher dimensions, but they won't necessarily include all possible energy assignments to vertex states. 

Because they determine the ground configuration, the coordinates $e_C,e_H,e_F$ are the most natural parametrization at low temperatures. In particular, the ground configurations, depicted in Fig.~\ref{fig:groundstates}, are those which consist only of the lowest-energy vertex states for that model. The full space of Hamiltonians is 3-dimensional, but two Hamiltonians related by a scaling factor represent the same model at different temperatures, so the space of all ``models" is a sphere. If we unwrap this sphere to the plane, we obtain a two dimensional diagram depicted in Fig.~\ref{fig:0temp}.

\begin{figure}
\centering
\begin{minipage}{0.6\linewidth}
\includegraphics[width=\linewidth]{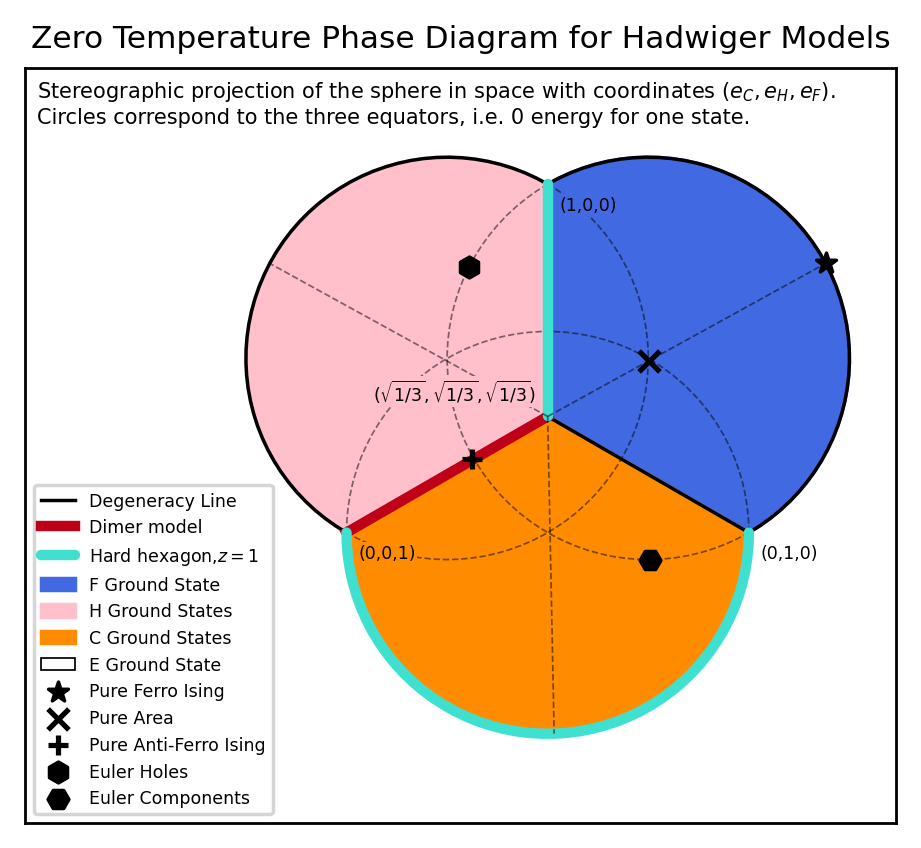}
\end{minipage}%
    \hspace{\fill}%
\begin{minipage}{0.4\linewidth}
\caption{At $T=0$ the behavior is determined only by states with minimal energy.  Along three of the degeneracy lines where multiple vertex states have minimal energy there is nonzero entropy, meaning infinitely many global states are minimal. When both $C$ and $H$ vertex states are minimal the allowed configurations are those of the $T=0$ antiferromagnetic Ising model, the same as those of the dimer model \cite{BH82}. When both $E$ and $C$ vertex states are minimal the allowed configurations are those of the hard hexagon model with the even distribution on allowed configurations \cite{B89}. The $F$ and $H$ line is symmetric to the $E$ and $C$ under a spin flip. ``Landmark points" are marked, including the pure Euler, Perimeter, and Area terms. The line containing ``Pure Ferro Ising," ``Pure Area", and ``Pure Anti-Ferro Ising" corresponds to the Ising model with external field.}
\label{fig:0temp}
\end{minipage}
\end{figure}

The diagram decomposes into regions where there is a unique vertex state with lowest energy, separated by degeneracy lines where two vertex states have equal and minimal energy. The $H$ region has three ground configurations, corresponding to the three sublattices of the hexagonal lattice, which maximize the number of holes. The $C$ region is the inverse of the $H$ region, having three ground configurations with maximally many components. The $F$ and $E$ regions have one ground configuration each, including all and none of the hexagons respectively.

\section{Main Results}

First, we will recall some terms so we can use them freely in the following sections. A Gibbs state of a model with Hamiltonian $\mathcal{H}$ is a mapping from bounded measurable functions $f$ on the set of configurations $\Omega$ to expectations $\langle f \rangle$ of those functions, such that for any finite region the expectations conditional on a particular boundary configuration $\rho$ are those generated by the Hamiltonian given $\rho$, as defined in Definition~\ref{def:hadwiger}. Note that this implicitly defines probability distributions on configurations for any finite regions via indicator functions. Such a Gibbs state is ``extremal" if it cannot be represented as a linear combination of other Gibbs states with positive coefficients. A configuration or distribution is ``translation-invariant" if it is invariant under all translations that preserve the lattice, while it is ``periodic" if it is invariant under a spanning sublattice. These terms are technically distinct, but any periodic model can be rendered translation-invariant by considering larger configuration spaces on tiles formed from the fundamental domain of the period lattice \cite{DS85}. A model satisfies the Peierls condition if there exists a finite collection of translation-invariant minimal-energy configurations and a constant $c$ such that, for any other configuration $\omega$ that differs from a ground configuration $\eta$ only on a finite set $\Gamma$,  $\HH(\omega)-\HH(\eta)\geq c|\Gamma|$.
\begin{definition}
A \textit{degeneracy line} refers to a curve such that multiple configurations, not related by a symmetry, have minimal energy, while a \textit{coexistence curve} is a curve on which there are multiple non-symmetric Gibbs states.
\end{definition}

With the space of models laid out, we now consider the low temperature behavior of the Hadwiger models. There is a chain of statements we must string together to restrict our Gibbs states. Dobrushin and Shlosman tell us that in two dimensions and at low temperatures, all Gibbs states of a translation-invariant (resp. periodic) Peierls model are translation-invariant (periodic)\cite{DS85}. Pirogov and Sinai tell us that given the Peierls condition, all translation invariant (or equivalently, periodic) Gibbs states are linear combinations of extremal ``pure phases" \cite{PS75}. Zahradnik, extending Pirogov--Sinai, tells us these ``pure phases" are dominated by a particular ground configuration of the model \cite{Z84}. Slawny and Bricmont tell us the position of the transition curves between areas such that different ground configurations have an extremal Gibbs state corresponding to them for a particular model \cite{BS89}\cite{S87}.

 \begin{definition}
 Given a ground configuration $G$ denote by $B(\omega)$ the set of vertices not agreeing with $G.$ We say that a Gibbs state is \textbf{``dominated"} by a certain global configuration $G$ --- or equivalently is a ``pure phase" --- if there exists  some constant $c$ such that the probability a vertex is in a cluster of vertices of perimeter $\Gamma$ not agreeing with $G$ is bounded above by $e^{-c\Gamma}$ \cite{Z84}. These pure phases are exactly those induced by choosing boundary conditions agreeing with that configuration, although they may be induced by other boundary conditions. To respect the symmetry of the model, if a set of $n$ configurations is symmetrically related, we instead mean that the model has a $1/n$ chance of being dominated by any of the $n$ pure Gibbs states, conditional on an even measure for all boundary conditions.
 \end{definition}
 
 For any point in the interior of a region such that a single vertex state has minimal energy (that is, the pink, orange, blue and white regions in Fig.~\ref{fig:0temp}), there will be a temperature low enough that the minimal configuration dominates. However, at any fixed low temperature, the lowest energy configuration may not dominate in the entire region. The asymptotic transition curves between areas where different phases dominate, along which more than one distinct asymmetric ground configuration dominates, can be calculated based on the technique of Slawny \cite{S87}\cite{BS89}. Using these techniques, we derive a phase diagram representing the low-temperature behavior of the Hadwiger models, depicted in Fig.~\ref{fig:lowtemp} and summarized in the following theorems:

 \begin{theorem}\label{thm:H&C}
At any point such that the unique lowest energy vertex configuration is $H$ or $C$, there exists $T_0$ such that for $T<T_0$ there are exactly three distinct extremal Gibbs states, dominated respectively by the three ground configurations, and every Gibbs state is a linear combination of these.
\end{theorem}

To identify the position of the coexistence curves relative to the degeneracy lines, we determine which state dominates along the degeneracy lines, which in turn tells us where the line and curve cross.

\begin{theorem}\label{thm:E-H}
For a fixed closed subset of the $E-H$ degeneracy line not containing its endpoints there exists $T_0$ such that, $\forall 0<T<T_0$, if $e_F>e_C$ the $E$ state dominates and if $e_F<e_C$ the $H$ state dominates. 
\end{theorem}
For example, along the closed subset $\min(E_C,E_F)\geq \epsilon$ of the degeneracy line $e_H=e_E=0$, there is a fixed $T_0$ such that the $E$ state dominates when $E_C<E_F$ and the $H$ state dominates when $E_C>E_F$. 
Moreover, the coexistence curve between the $E$ and $H$ states crosses this degeneracy line at the point $e_F=e_C=\frac{\sqrt{2}}{2},e_H=e_E=0$,  passes inside the region where $H$ is the ground state when $e_F>e_C,$ and is contained inside the region where $E$ is the ground state when $e_F<e_C.$ 

The corresponding statement under a spin flip also holds for the $C-F$ degeneracy line.

\begin{theorem}\label{thm:E-F}
For a fixed closed subset of the $E-F$ degeneracy line not containing its endpoints there exists $T_0$ such that, $\forall 0<T<T_0$,  if $e_H>e_C$ the $E$ state dominates and if $e_H<e_C$ the $F$ state dominates. 
\end{theorem}

For the non-Peierls $E-C$ and $E-H$ degeneracy lines, we can't use Pirogov-Sinai, so the following theorem is proven using disagreement percolation and translation invariance of the unique Gibbs state. 
\begin{theorem}\label{thm:E-Cb}
Along the $E-C$ degeneracy line, if $e_F\geq e_H$ then there is a unique Gibbs state, for which no configuration dominates at any temperature. The same is true along the $H-F$ degeneracy line when $e_E\geq e_C$.
\end{theorem}
We also obtain a weaker analogue of domination as an application of reflection positivity, which all Hadwiger models satisfy.
\begin{theorem}\label{thm:E-Cc}
Along the $E-C$ degeneracy line, if $e_F\geq e_H$ then probability a vertex is in either the $H$ or $F$ states decays exponentially in the inverse temperature. The same is true for $E$ and $C$ states along the $H-F$ degeneracy line where $e_E\geq e_C$.
\end{theorem}

 \begin{figure}
\centering
\begin{minipage}{0.75\linewidth}
\includegraphics[width=\linewidth]{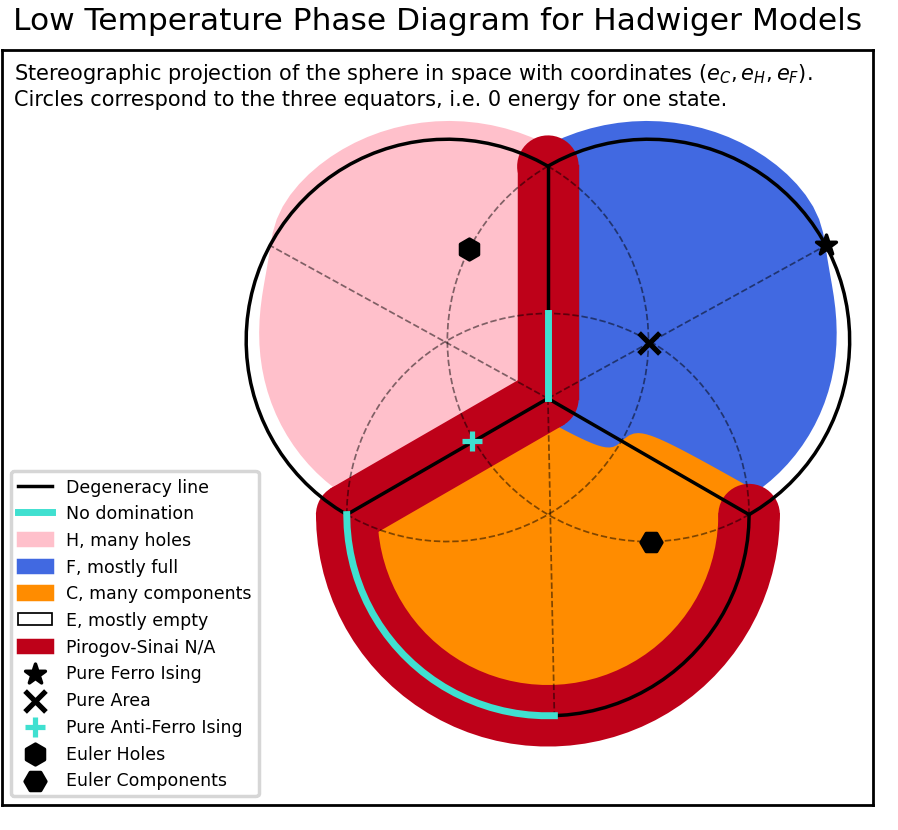}
\end{minipage}%
\begin{minipage}{0.25\linewidth}
    \caption{Phase diagram showing the regions where given configurations dominate. Dark regions surrounding the non-Peierls degeneracy lines indicate areas where Pirogov--Sinai techniques are not applicable. This region shrinks with decreasing temperature, but is never empty. Outside of those areas each color, including white, represents a region where a particular set of configurations dominates. On the boundaries between these regions multiple Gibbs states not related by a symmetry dominate. On some of the non-Peierls lines we use disagreement percolation to prove no domination at any temperature.}
    \label{fig:lowtemp}
\end{minipage}
\end{figure}

\section{Proofs of the the Main Theorems}

Before we make any specific arguments, first we should note the symmetry of the model. Flipping all the spins interchanges $E$ vertices for $F$ vertices, $C$ vertices for $H$ vertices, and vice versa. Thus, flipping the spins while swapping the energies assigned to the respective vertex states maintains the assignment of probabilities. With this in mind, any statement that concerns particular vertex states applies just as well to its ``twin" under spin-flip symmetry. We will frequently make an argument for one state, then apply the same to its twin by spin-flip symmetry.

Using techniques from Pirogov--Sinai, we address the large regions dominated by a particular vertex state first, in subsection \ref{sec:PS}. Next, with Slawny's technique we handle the position of the coexistence curves between these regions in subsection \ref{sec:CL}. Finally, with disagreement percolation and reflection positivity we determine the behavior along the degeneracy lines for which Pirogov--Sinai techniques are inapplicable, in subsection \ref{sec:NPL}. 

\subsection{Applying Pirogov--Sinai}
\label{sec:PS}

We begin by describing the ground states of the model, which define the regions in Fig.~\ref{fig:0temp}. Given two vertex states $A,B$, we will refer to the points where $e_A=e_B$ and both achieve the minimal vertex energy as the $A-B$ degeneracy line.

\begin{lemma}
Along the $H-F$, $E-C$, or $H-C$ degeneracy lines, there are infinitely many ground configurations. At all other points, there are finitely many ground configurations, all of which are periodic. 
\end{lemma}

\begin{proof}

First, we consider the case where exactly one vertex state has minimal energy. If a configuration with only that vertex state is possible, it will be minimal. For the $F$ and $E$ vertex states, minimal energy is achieved by the full and empty configurations respectively; any other configuration must have other vertex states, and so the ground configuration is unique and translation invariant. For the $H$ and $C$ vertex states, we start with a single vertex. This could have three distinct assignments, as the $C$ state is realizable by three different three-hexagon configurations. However, given a trio of three hexagons and state $C$, there is only one configuration on its neighbors such that every vertex is in state $C$. By induction, this uniquely determines a periodic ground configuration given one of our original choices, so there are three periodic ground configurations. Along the $E-H$, $E-F$, and $C-F$ degeneracy lines, the minimal vertex states cannot coexist, i.e. neighboring vertices cannot have distinct minimal vertex states, they must be separated by a third vertex that is in a non-minimal state, as demonstrated in Fig.~\ref{fig:interfaces}. Therefore, all minimal configurations must have only one vertex state, so the set of minimal configurations is the union of the set on the two neighboring regions, and as such is finite and periodic.

\begin{figure} [b]
\centering
\begin{minipage}{0.75\linewidth}
\includegraphics[width=\linewidth]{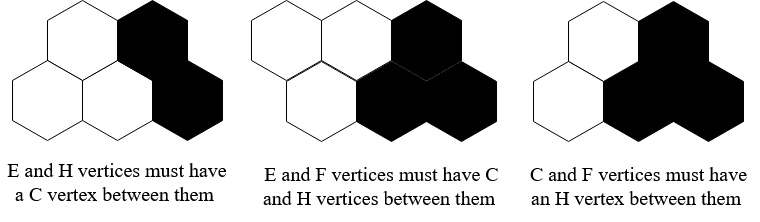}
\end{minipage}%
\begin{minipage}{0.25\linewidth}
    \caption{Each of the three pairs of vertex states, $E-H, {E-F}, C-F$ which cannot be directly adjacent, but must have a third vertex state between them}
    \label{fig:interfaces}
\end{minipage}
\end{figure}

Along the $H-F$, $E-C$, and $H-C$ boundary lines, minimal vertex states can coexist. Thus there are infinitely many energy-preserving modifications to particular ground configurations (for example, given the empty configuration, along the $E-C$ line, one can fill in any single hexagon), so the set of ground configurations is infinite. These ground configurations are also not in general periodic: the lower half-plane may have one vertex state, while the upper may have the other. 

\end{proof}

For Pirogov--Sinai arguments to apply, we need the model to satisfy the Peierls condition. 

\begin{lemma}All Hadwiger models satisfy the Peierls condition, except on the $H-F$, $E-C$, or $H-C$ degeneracy lines.\end{lemma}
\begin{proof}
An $m$-potential is a function on $\omega\in \{-1,1\}^S$ that can be written as $\sum_{U\subset S}f_U(\omega)$, where each $f_U$, called a ``potential", depends only on elements in $U$, and there exists a configuration $\omega$ such that $f_U$ achieves its minimal value for all $U$. As described above, any Hadwiger Hamiltonian can be represented as the sum over the vertices of terms that depend only on the neighboring hexagons, and may have every vertex in a minimum-energy state, so the Hadwiger models are $m$-potentials. If an $m$-potential's set of minimal configurations is finite, it satisfies the Peierls condition \cite{FV17}. As just demonstrated, outside the $H-F$, $E-C$, or $H-C$ degeneracy lines the set of minimal configurations is finite, so at those points the Peierls conditions is satisfied. 

Along the $H-F$, $E-C$, and $H-C$ degeneracy lines, there are infinitely many ground configurations; these are achieved by starting from one ground configuration and flipping any of the infinitely many hexagons that do not increase the energy. Thus, they do not satisfy the Peierls condition and have nonzero entropy at 0 temperature. In particular, the ground states of $H-F$ and $E-C$ are the allowed states of the hard hexagon model, while the ground states of the $H-C$ model are ground states of the frustrated antiferromagnetic Ising model, or equivalently, random lozenge tilings \cite{G21}.
\end{proof}

The low temperature behavior in the large two-dimensional regions of the phase diagram is a straightforward consequence of Pirogov--Sinai Theory.

\begin{proof}[\textbf{Proof of Theorem \ref{thm:H&C}}]
We can introduce two periodic but not rotationally invariant ``dummy fields" $f_i$ that distinguish between the three forms of the $C$  vertex state, $\{C_i\}$, by $f_i(C_j)=\delta(i,j)$. Note that we only need two fields, as up to a constant the third can be represented implicitly by negative energy on the other two states. When these fields have nonzero coefficient, they lift the degeneracy of the three configurations, so by Pirogov--Sinai \cite{PS75}, $\exists T_0$ such that $\forall T<T_0$, for each form there is a choice of parameters where this forms dominates, and that at their coexistence point there is an equal chance of any dominating. By the symmetry of three ground configurations, we know this coexistence point must be at the point where both dummy fields are 0, i.e. the original model.

Zahradnik \cite{Z84} states for any periodic Peierls model with finite ground configurations $\exists T_0$ such that $\forall T<T_0$, every periodic Gibbs state is a linear combination of extremal Gibbs states dominated by a periodic ground configuration. Dobrushin and Shlosman \cite{DS85} state that for 2-dimensional models, under the same conditions, every Gibbs state is periodic. Therefore, for any 2-dimensional periodic Peierls model with finite ground configurations $\exists T_0$ such that $\forall T<T_0$, every Gibbs state is a linear combination of those dominated by the ground states. The model always satisfies the other conditions, so at every Peierls point at sufficiently low temperatures, the space of Gibbs state is exactly those dominated by the ground states. 

\end{proof}

\subsection{Coexistence Curves}\label{sec:CL}

At fixed low temperature, there are three coexistence curves to consider, the latter two of which are equivalent under spin-flip symmetry: $E-F$, $E-H$, and $C-F$. Zahradnik tells us that coexistence curves approach the zero-temperature degeneracy lines, but does not tell us from which direction they approach it, or whether they achieve it at some nonzero temperature. There is a technique by Slawny \cite{S87} to calculate the asymptotic low-temperature positions of these coexistence curves explicitly, by calculating at each point which of the ground configurations has the lowest-energy or highest-multiplicity minimal excitation. We use this technique in a slightly unconventional way.  Normally Slawny's technique is used to determine the location of a low-temperature triple point, where three coexistence curves meet, relative to a Peierls zero-temperature triple point. No Peierls triple point exists here, but we can still determine the relative location of the coexistence curve by calculating 1-dimensional ``slices" perpendicular to each point along the zero-temperature degeneracy line. In each slice, there is a nonzero $T$ such that the transient terms are smaller than the asymptotic term in the expansion we describe below. From Slawny we know the terms vary continuously, so on a closed region there will be a nonzero minimum, such that the overall curve (which by implicit function theorem must be continuous) approximates the asymptotic curve.

\begin{proof}[\textbf{Proof of Theorem \ref{thm:E-H}}]

Along this line, the $E$ and $H$ vertex states have energy 0, and moving along the line changes the energy of $F$ and $C$. The space of models lies on a sphere, but we're interested only in the asymptotic shape and position relative to the degeneracy lines, so we can change coordinates at will as long as it does not affect relative position. Project to a plane to simplify coordinates, so along the degeneracy line $e_F=f\in [\epsilon,1-\epsilon]$ and $e_C=1-f$. Moving perpendicular to this line, $e_H=h$. Following Slawny, we can treat the three symmetric ground configurations as a single ground configuration for the purposes of the calculation.

Slawny's technique consists of solving the equation
\begin{equation}
\dot{P}^G(\beta \mathcal{H}_0+\mathcal{H})-e_G(\mathcal{H}')=\dot{P}^{G'}(\beta \mathcal{H}_0+\mathcal{H}')-e_{G'}(\mathcal{H}')
\end{equation}
where $\dot{P}^G$ is the pressure of the gas of excitations from the ground configuration $G$ with respect to a given Hamiltonian, $\mathcal{H}_0$ is the Hamiltonian you're expanding around, $\mathcal{H}'$ is the perturbation to the Hamiltonian, and $e_G$ is the energy density of $G$ with respect to a given Hamiltonian, equivalent to the energy the Hamiltonian assigns to the associated vertex state. By ``pressure of the gas of excitations", we mean $\log(\sum_{m\in M} e^{-\HH(M)})$, where $M$ is the set of configurations that disagree with your ground configuration at only finitely many points, and have energy below your chosen threshold. Such a solution identifies the points where the modified pressures of the two ground states are equal, and so both have corresponding Gibbs states: the coexistence curve. This equation is difficult to solve explicitly, so we instead solve up to the leading term in the expansion in $e^{-\varepsilon_i}$, where $\varepsilon_i\in \mathcal{E}$, the set of possible excitation energies. Denote equation up to leading term by $\approx$. We find the coefficients using the cluster expansion. Because we only consider the leading order term, which consist of single polymers, the coefficient in $\dot{P}^G$ simplifies to the density of minimal excitations. In the following calculations, we will use $\beta$ for $1/T$, to simplify the notation. 

The $E$ configuration only has one kind of minimal excitation, flipping a single hexagon. This changes 6 vertices in state $E$ to 6 vertices in state $C$, changing the energy by $1-f$. The perturbation of the Hamiltonian assigns $E$ 0 energy, so that term disappears.

The $H$ configuration has one of two possible minimal excitations, depending on the choice of parameters.  Flipping one of the empty hexagons which form 1/3 of the configuration replaces 6 $H$ vertices with $F$ vertices, changing the energy by $6(f-h)$. Flipping a full hexagon replaces 6 $H$ vertices with $C$ vertices, changing the energy by $6(1-f-h)$. Combining this information, we solve 
\begin{equation}
e^{6(1-f)\beta}\approx \frac{1}{3}e^{6f\beta-6h}+\frac{2}{3}e^{6(1-f)\beta-6h}+h.
\end{equation}
We set $h=0$ in the exponent, because when $h<<A$, $e^{A+h}-e^A$ is lower order than $e^A$. Thus, we have
\begin{equation}h\approx \frac{1}{3}e^{6(1-f)\beta}-\frac{1}{3}e^{6f\beta}.\end{equation} 
Therefore $h>0$ when $f<1/2$ and $h<0$ when $f>1/2$.

At $f=1/2$, there is a symmetry between the $E$ and $H$ excitation expansions. For any minimal excitation, flipping the value of hexagons in two sublattices replaces vertices in state $E$ with $H$ and vice versa, while vertices in state $F$ are replaced with $C$, again vice versa. Therefore if $e_H=e_E$ and $e_F=e_C$, the multiplicities and energies of excitations around $E$ and $H$ will be the same, the pressures will be the same, and so by symmetry the curve of the phase transition must pass through $f=1/2$ at $h=0$. This implies the coexistence curve passes through the point at $f=1/2$, with the $E$ phase extending over the line when $f>1/2$ and the $H$ phase extended over the line when $f<1/2$.

The behavior along $C-F$ curve is the same, by spin-flip symmetry.

\end{proof}

\begin{proof}[\textbf{Proof of Theorem \ref{thm:E-F}}]
Again we project to a plane, $F$ and $E$ both have energy $0$, $e_h=h\in [-1+\epsilon,1-\epsilon]$, and $e_C=1-h$ in $[-1+\epsilon,1-\epsilon]$. Moving perpendicular to this line, $e_F=f$. In the $E$ configuration, flipping an empty hexagon replaces 6 $E$ vertices with $C$ vertices, changing the energy by $1-h$, while in the $F$ configuration, flipping  a full hexagon replaces 6 $F$ vertices with $H$ vertices, changing the energy by $h-f$. Thus, we solve
\begin{equation} e^{(1-h)\beta}\approx e^{h\beta-f}+f \end{equation}
Again setting the varying element to 0 in the exponent, we get
\begin{equation}f\approx e^{(1-h)\beta}-e^{h\beta}\end{equation}
so $f>0$ when $h<1/2$ and $f<0$ when $h>1/2$. By the symmetry of the two ground states, we know this curve must be odd around the equality point, and thus passes exactly through 0 there.
\end{proof}

\subsection{Non-Peierls Lines}\label{sec:NPL}

Now we handle two of the three non-Peierls infinitely degenerate degeneracy lines. We can't use Pirogov--Sinai here, so we will instead use disagreement percolation to prove uniqueness, then apply reflection positivity to understand the behavior. Our argument relies on the following uniqueness condition \cite{BM94}:
\begin{theorem}[Disagreement Percolation]\label{disperc}
Fix a particular Markov field, and denote $N_i$ as the set of neighbors of a single site, $d(\cdot,\cdot)$ as the variational distance, and $Y_i(\cdot, \eta)$ as the one-site distribution given its neighbors. If $p_i<p_c$, where $p_i=\max_{\eta,\eta'\in \{-1,1\}^{N_i}}d(Y_i(\cdot,\eta),Y_i(\cdot,\eta')$ and $p_c$ is the critical site percolation threshold for the given lattice, then there is exactly one Gibbs measure induced by the Markov field.
\end{theorem}

\begin{corollary}\label{thm:E-Ca}
Along the $E-C$ degeneracy line, if $e_F\geq e_H$, there is a unique Gibbs state at all temperatures. By spin-flip symmetry, the same is true along the $F-H$ line where $e_E\geq e_C$.
\end{corollary}
\begin{proof}

When $e_F\geq e_H$, filling the central hexagon never increases the energy, so it always has probability greater or equal to $1/2$ regardless of the boundary conditions. The probability of a filled hexagon is always less than 1, so the variational distance between any two boundary conditions is less than 1/2, the critical site percolation threshold for the triangular lattice \cite{K82}, and so we have uniqueness for all $T$. When $e_F<e_H$, then the maximum variational distance approaches 1: for entirely filled boundary conditions an empty central hexagon is preferred, but for boundary conditions with a single filled face, a filled central hexagon is preferred.

\end{proof}

\begin{figure}[t]
    \centering
    \includegraphics[width=\linewidth]{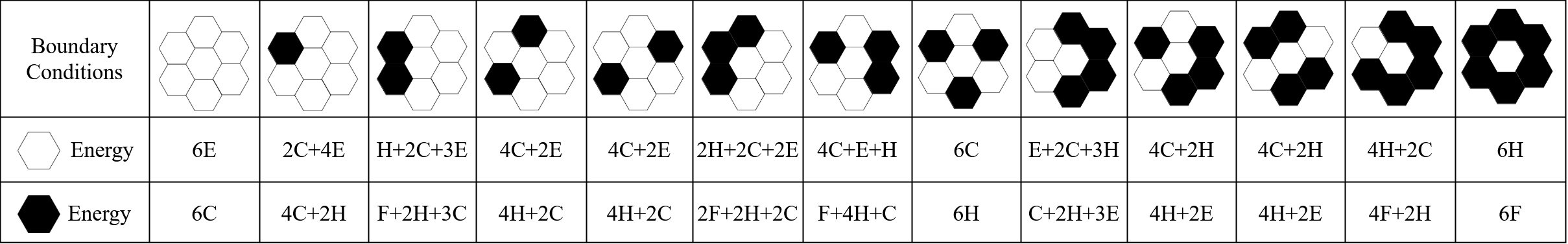}
    \caption{Quantities of the vertex states given either face assignment for each boundary condition, up to rotation and reflection. With this, the relative probability of the face assignments can be determined for any boundary condition and energy assignment, to determine the variational distance between induced probability distributions for disagreement percolation.}
    \label{fig:boundaries}
\end{figure}

\begin{proof}[\textbf{Proof of Theorem \ref{thm:E-Cb}}]
From  Corollary~\ref{thm:E-Ca} above, we know that the Gibbs state is unique. If a Gibbs state is unique, it must also be translation-invariant. Therefore, any dominating configuration must also be translation invariant. However, the only translation-invariant configurations are $E$ and $F$. The $F$ configuration has maximal energy and so is minimally likely among all configurations. If we assume the $E$ configuration dominates then the probability of any face being in state $E$ must go to unity as $T\to 0$. However, any face whose neighbors are all in state $E$ can have its value flipped without increasing the energy, so that configuration must be equally likely as the $E$ configuration, conditional on empty neighbors. Thus, the probability of a filled hexagon cannot go to 0: it must be no less than $1/2$ the probability that all its neighbors are empty. If the probability that all its neighbors are empty goes to 0, the probability of a filled neighbor cannot go to 0, and we have a contradiction by translation invariance.
\end{proof}

Although we lack domination, at low temperatures the set of allowed configurations is still constrained with high probability. However, instead of approximating a particular configuration, they are constrained to avoid particular vertex states. Because the Gibbs measure is unique and the space of measures is sequentially compact, the unique infinite Gibbs measure is achieved by any sequence of local measures with infinite limiting radius. Thus, we can restrict our attention to measures defined on the rectangular torus. A reflection is defined by a partition of the torus into two halves, by a pair of horizontal or vertical lines, such that reflecting across the lines preserves the lattice. We will use the ``closed half-space" to refer to the set of faces intersecting of one of these two halves. The particular reflections we're considering here are represented in Fig.~\ref{fig:reflections}.

\begin{figure} [b]
\centering
\begin{minipage}{0.65\linewidth}

    \includegraphics[width=\linewidth]{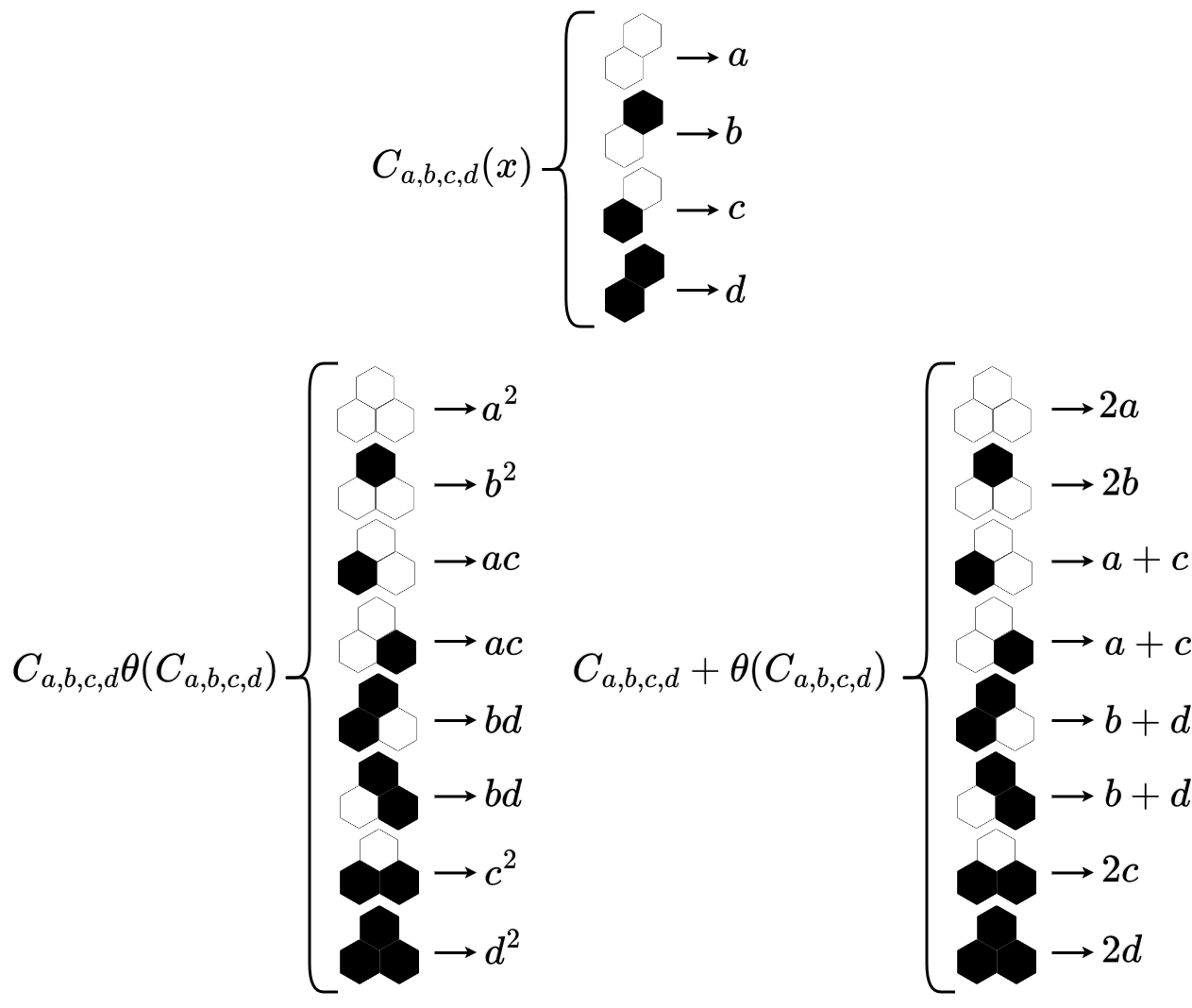}
    
\end{minipage}%
    \hspace{\fill}%
\begin{minipage}{0.3\linewidth}
    \caption{(left)Local energy functions considering two boundary hexagons, and the induced functions on three boundary hexagons, two in each half space, with the top hexagon in both closed half-spaces. With these we construct indicator functions for each vertex state}
    \label{fig:Cabcd}
    \vspace*{10pt}
    \includegraphics[width=\linewidth]{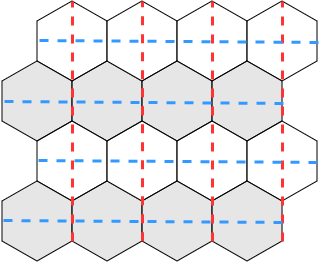}
    \caption{Horizontal and vertical reflections define blocks containing one full vertex}
    \label{fig:reflections}
\end{minipage}

\end{figure}

\begin{theorem} All Hadwiger models on the torus satisfy reflection positivity, i.e if $\theta$ is a horizontal reflection through faces or a vertical reflection through faces and vertices, and $f,g\in A^+$, the set of functions on the closed half-space of the reflection, $\langle f\theta (g)\rangle=\langle g\theta (f)$ and $\langle f\theta (f)\rangle\geq 0$ \end{theorem}

\begin{proof}

The first equation follows immediately from $\theta$ invariance of the Hamiltonian under both kinds of reflections. For horizontal reflections through the faces, the locality of the model immediately implies the second relation holds; see the argument in \cite{HP22}. For vertical reflections through sites and bonds, we use the following statement: if the Hamiltonian can be realized as $A+\theta (A)+\sum_i C_i\theta (C_i)$, with $A, C_i\in A^+$, the induced measure is reflection positive \cite{FV17}. Each vertex internal to one of the half-spaces has an associated potential $A\in A^+$, and a corresponding vertex in the other half with the potential in $\theta(A)$. However, the vertices along the border do not have potentials in $A^+$, so representing the associated potential in these terms is nontrivial. To achieve this, we parametrize all functions $C_{a,b,c,d}\in A^+$ that depend only on a pair of boundary hexagons by the energy they assign to each of the four possible configurations. Then, we combine them to look at the associated functions on triplets in $A^+\theta( A^+)$ and $A^++\theta(A^+)$. These functions are represented visually in Fig.~\ref{fig:Cabcd}.

We can construct an indicator function for the $E$ and $F$ vertex states with $C_{1,0,0,0}\theta (C_{1,0,0,0})$ and $C_{0,0,0,1}\theta (C_{0,0,0,1})$ respectively. For the $H$ vertex state we can construct an indicator functions with \\${C_{1,1,1,0}\theta( C_{1,1,1,0})+C_{-1/2,0,-1/2,0}+\theta (C_{-1/2,0,-1/2,0})}$, and similarly for the $F$ vertex state. Thus any Hadwiger energy function, represented in terms of relative energy values of vertex states, leads to a reflection positive measure on the torus. 
\end{proof}

With reflection positivity we are able to use the chessboard estimate. Define a ``block" as the fundamental domain of the set of reflections including every horizontal reflection and every other vertical reflection, displayed in Fig.~\ref{fig:reflections}. This will include one vertex in its center, and two along its edge. For a local event $A$ on a single block $B$, we can define an event $A_i$ on any other block $B_i$ by applying the necessary reflections to map that block to $B$. Denote by $\alpha$ the global event that $A_i$ occurs simultaneously on all blocks, indexed by $S$. The chessboard estimate states \cite{FV17}:
$\PP(A)\leq \PP(\alpha)^{1/|S|}\,.$

\begin{proof}[\textbf{Proof of Theorem \ref{thm:E-Cc}}]

We decompose the event into $\sigma= F$ and $\sigma=H$. First, consider $\sigma=F$. In this case, if $\sigma_i=F \forall i\in S$, then every hexagon is filled, so $\sigma_i=F, \forall i$, which corresponds to exactly 1 configuration. The total number of vertices is twice the number of blocks, so the weight of this configuration is $e^{-2e_F|S|/T}$ and the partition function of the whole system is bounded below by 1, so 
$$\PP(\sigma=F)<(e^{-2e_F|S|/T})^{1/|S|}=e^{-2e_F/T}.$$ Now we consider $\sigma=H$. There are no more than $3^{|S|}$ configurations such that ${\sigma_i=H} ,\,\forall i\in S$, because each block can take no more than 3 independent configurations with state $S$. The non-central vertices may be in a lower energy state, but the energy of the total configuration must be at least $e_H|S|$, so 
$$\PP(\sigma=H)<(3^{|S|}e^{-e_F|S|/T})^{1/|S|}=3e^{-e_F/T}.$$ Thus $\PP(\sigma=H\text{ or }F)<e^{-2e_F/T}+3e^{-e_F/T}$. The same statement holds along the $H-F$ degeneracy line, when $e_E\geq e_C$.
\end{proof}
With this, we can make the connection to the hard hexagon model explicit at sufficiently low temperatures. 
For a fixed infinite configuration $\eta$, denote the probability of an event $A$ for the subcritical (i.e. no domination) Hadwiger model at temperature $T$ on the finite region $R$ with boundary conditions induced by $\eta$ as $\PP^R_{S,\eta,T}(A)$ and the distribution of the Hard Hexagon model with $z=1$ (i.e. even distribution) on the same region with the same boundary conditions by $\PP^R_{H,\eta}$. Denote their respective infinite volume limits by $\PP_{S,\eta,T}(A)$ and $\PP_{H,\eta}$.

\begin{theorem}
For a fixed infinite boundary condition $\eta$ and local event $A$ depending on the set $\lambda$, $\lim_{T\to 0} \PP_{S,\eta,T}(A)=\PP_{H,\eta}(A)$
\end{theorem}
\begin{proof}
Without loss of generality, assume that $\lambda$ is an $N\times N$ rectangle. Corollary 1 of \cite{BM94} states:
$$
|\PP^M_{S,\eta_1,T}(A)- \PP^M_{S,\eta_2,T}(A)|\leq \PP_{p_i}(\text{there is an open path from some vertex in } N \text{ to } \partial M)
$$
where $\PP_{p_i}$ is Bernoulli site percolation with probability $p_i$, and $p_i$ is the variational distance between neighbor boundary defined in Theorem~\ref{disperc}. $p_i<1/2$ at any non-zero temperature, so the probability of an open path is always less than at the critical probability, $1/2$. For fixed $N$, we can choose $M_0$ such that \cite{K82} 
$$
\PP_{1/2}(\text{there is an open path from some vertex in } N \text{ to } \partial M)<\epsilon/2 $$
Denote by $A^M_{EC}$ the event that no vertex in $M$ is in the state $F$ or $H$. Given $M_0$, we can choose $T$ such that $\PP^M_{S,\eta,T_0}(A^M_{EC})>1-\epsilon/2$, by the previous theorem. $\PP^M_{H,\eta}(A)=\PP^M_{S,\eta,T}(A|A^M_{EC})$, for any $T$, by the definition of the hard hexagon model, so 
$$
|\PP^M_{S,\eta,T}(A)-\PP^M_{H,\eta}(A)|=\big|\PP^M_{S,\eta,T}(A|A^M_{EC})-\PP^M_{S,\eta,T}(A|\neg A^M_{EC})\big|\cdot\PP^M_{S,\eta,T}(\neg A^M_{EC})\leq \epsilon/2
$$
Now, consider a larger $M'$ containing $M$. 
$$
\min_{\eta'}\PP^{M}_{S, \eta',T}(A)\leq \PP^{M'}_{S, \eta,T}(A)\leq \max_{\eta'}\PP^{M}_{S, \eta',T}(A)$$
$$
\min_{\eta'}\PP^{M}_{H, \eta'}(A)\leq \PP^{M'}_{H, \eta}(A)\leq \max_{\eta'}\PP^{M}_{H, \eta'}(A)
$$
because $\eta_{M'}$ affects $A$ only by inducing a probability distribution on $\eta_M$, by the Markov property. Thus

\begin{align*}
|\PP^{M'}_{S, \eta,T}(A)-\PP^{M'}_{H, \eta}(A)|&\leq \max_{\eta_1',\eta_2'}|\PP^{M}_{S, \eta_1',T}(A)-\PP^{M}_{H, \eta_2'}|\\
&= \max_{\eta_1',\eta_2'}|\PP^{M}_{S, \eta_1',T}(A)-\PP^{M}_{S, \eta_2',T}(A)+\PP^{M}_{S, \eta_2',T}(A)-\PP^{M}_{H, \eta_2'}|\\
&\leq \max_{\eta_1',\eta_2'}|\PP^{M}_{S, \eta_1',T}(A)-\PP^{M}_{S, \eta_2',T}(A)|+|\PP^{M}_{S, \eta_2',T}(A)-\PP^{M}_{H, \eta_2'}| \\
&\leq \epsilon/2+\epsilon/2=\epsilon
\end{align*}
We can always choose $T$ such that $|\PP_{S,\eta,T}(A)-\PP_{H,\eta}(A)|\leq \epsilon$, because the inequality holds for every sufficiently large $M'$, and so $\lim_{T\to 0} \PP_{S,\eta,T}(A)=\PP_{H,\eta}(A)$
\end{proof}

\begin{corollary}
For all $A$, $\PP_{H,\eta}(A)$ is unique regardless of $\eta$, so the even-distribution hard hexagon model has a unique infinite volume limit. 
\end{corollary}
 \begin{proof}
$\lim_{T\to 0} \PP_{S,\eta,T}(A)=\PP_{H,\eta}(A)$, and we already established that $\PP_{S,\eta,T}(A)$ is independent of $\eta$, so $\PP_{H,\eta}(A)$ is independent of $\eta$. Note that for $z<1$, the result follows immediately from disagreement percolation.
 \end{proof}

\section{Special Cases}

To our knowledge, the only finite-energy model in the Hadwiger class other than the Ising model in a field that has been exactly solved is the Baxter-Wu model, sometimes called the Ising model with three-spin interactions \cite{BW73}. This has Hamiltonian $-\sum_v h_{v1}h_{v2}h_{v3}$, where $h_{v1},h_{v2},h_{v3}$ are the spins in $\{-1,1\}$ of the hexagons adjacent to $v$, which results in a model at the midpoint of the $C-F$ or $E-H$ degeneracy lines, depending on sign convention. This makes it a natural analogue of the Ising model in the configuration space, which is the midpoint of the only other Peierls degeneracy line. The Baxter-Wu model exhibits a sharp phase transition, which, given the sharp transitions in the Ising model as well,  is reason to believe that in fact all Hadwiger models exhibit sharp phase transitions separating the different regions of the diagram, except along the non-Peierls degeneracy lines. 

At the pure Euler point, the Hamiltonian has a number of very nice properties, particularly on the hexagonal lattice. Euler Characteristic is the difference of components and holes, so with fixed boundary conditions there is an isomorphism between the pure Euler model and a model defined on loop configurations $\omega$, with energy $O(\omega)-I(\omega)$, with $O(\omega)$ and $I(\omega)$ denoting the number of loops contained within an even and odd number of loops respectively (generally loops will not be deeply nested, so this corresponds to ``outer" and ``inner" loops). Adding the perimeter term results in a local analogue of the loop $O(n)$ model. At low temperatures, where nested loops are very rare, these models display similar behavior, and the loop $O(n)$ results can be used to make an alternate proof of three distinct Gibbs states\cite{DC17}. However, due its nonlocality the loop $O(n)$ model is much less amenable to traditional techniques, and so does not provide the completeness results of Pirogov--Sinai-Zahradnik.

\section{Further Directions}

As any finite-range model displays only a single undominated Gibbs state at sufficiently high temperatures \cite{FV17}, models in the $C$ and $H$ regions must experience a change in the number of Gibbs state at some intermediate temperature. However, precisely where this change occurs, and if there are temperatures with a different number of Gibbs states, is yet unknown. For any set of parameters not on the non-Peierls degeneracy lines, a transition occurs in the sense that some configurations dominate at low temperature, while at high temperatures no configuration dominates. However, the nature of this transition is also not known.

This paper relies almost entirely on very generic features of the Hadwiger models. However, the Hadwiger property is a very strong restriction on the space of 2-dimensional models, and should make this class amenable to much stronger results using these subtler properties. With face assignments of $\{1,-1\}$, the area term is $\sum_i x_i$, the perimeter term is $\sum_{(i,j)\in e}x_ix_j$ and the Euler term is $\sum_{(i,j,k)\in v}x_i+x_j+x_k-3x_ix_jx_k$, so all Hadwiger models are generalized Ising models, as mentioned in the introduction. However, any model with nonzero Euler term is not ferromagnetic, so results about the general class of ferromagnetic Ising models fail to hold. 

The structure of strictly locally geometric models depends significantly on the underlying lattice. As the maximally symmetric 2D infinite graph, the hexagonal graph provides the simplest dynamics, but the model can also be considered on other lattices. On the usual square lattice, Hadwiger models do not span the set of all vertex state assignments, because there are 4 nonenmpty assignments and only 3 Hadwiger basis vectors. Thus in that domain the space of ground states is not nearly as evident, nor is there an obvious natural choice of basis. 

In three dimensions, the space of invariant valuations is four-dimensional, being parametrized by Euler characteristic, mean width, surface area, and volume \cite{G07}. Here, mean width refers to the average projected length of the set over the space of 1-dimensional subspaces. The natural analogue to the lattice of hexagons is the lattice of truncated octahedra; for both, $d+1$ identical cells meet at each vertex. However, unlike the hexagonal lattice, for which the Euler characteristic is antisymmetric, the Euler characteristic is now symmetric, paired with the symmetric surface area term. Volume remains antisymmetric, along with the new mean-width terms. Up to temperature normalization, this generates two 1-dimensional subspaces of purely symmetric or anti-symmetric energy functions, granting the model many more potential symmetries to rely on. 

\section*{Acknowledgments and Declarations}

Acknowledgments: We'd like to thank Alexander Smith at the University of Minnesota as well as {\'E}rika Rold{\'a}n and Ron Peled for essential conversations on the paper. In particular, Smith noticed the connection between the class of invariant valuations and the Ising Hamiltonian, which inspired the whole project. We'd also like to thank the anonymous referees for their useful comments and suggestions.\\*
\\*
Competing Interests: We have no competing interests to declare.\\*
\\*
Data Availability Statement: No datasets were generated or analyzed during the current study.

\bibliography{refs}

\backmatter



\end{document}